\documentclass[11pt,a4paper]{article}
\usepackage{amsmath,amssymb,amsthm}
\usepackage{booktabs}
\usepackage{graphicx}
\usepackage{hyperref}
\usepackage[margin=2.5cm]{geometry}
\usepackage{enumitem}
\usepackage{microtype}

\newtheorem{theorem}{Theorem}
\newtheorem{lemma}{Lemma}
\newtheorem{proposition}{Proposition}
\newtheorem{definition}{Definition}
\newtheorem{remark}{Remark}
\newtheorem{observation}{Observation}
\newtheorem{assumption}{Assumption}

\title{Vote-Left: A Deterministic Coordination Strategy for the Faithful
in The Traitors}
\author{Vincent Knight \\ \texttt{knightva@cardiff.ac.uk}}
\date{\today}

\begin{document}

\maketitle

\begin{abstract}
\emph{The Traitors} is a social deduction game in which an informed minority
of Traitors face an uninformed majority of Faithful, and the recurring
question facing the Faithful is how to vote. Eliminating a player chosen
uniformly at random is optimal for the Faithful when it can be guaranteed, but
when votes are cast individually a random vote is indistinguishable from a
strategic one, leaving the Faithful exposed to coordinated Traitor collusion.
We introduce the \emph{Vote-Left protocol}, a deterministic rule under which
every player votes for the next surviving player in a fixed cyclic ordering.
It reproduces the banishment distribution of random voting, yet any deviation
from it is immediately identifiable. Combined with a simple punishment rule,
Vote-Left punishes every deviation in the main phase of the game that does not
end it at once, whether by a single player or by the Traitors acting together.
In the final rounds punishment arrives too late, and the Faithful should return
to random voting there. Under this switching protocol no player gains by
departing from Vote-Left in the main phase, and with a single Traitor it is a
Perfect Bayesian Equilibrium of the whole game. Across the televised configurations,
it raises the Faithful's winning probability by a factor of approximately four
over random voting under collusion with three Traitors, and by a factor of 62
to 84 with four.
\end{abstract}

\section{Introduction}

Social deduction games pit an informed minority against an uninformed
majority, and have become a popular setting for studying voting,
communication, and the use of public information under uncertainty. The
classical example is the Mafia game \cite{davidoff1986}, in which Mafia
members covertly eliminate citizens at night and, by day, the full group
votes to eliminate a suspect. \emph{The Traitors} (originally \emph{De
Verraders}, Netherlands, 2021) adapts this structure for television: a host
assigns a small group as Traitors (who know one another), the remainder
are Faithful, and play alternates between night murders and day
banishment votes until either all Traitors have been eliminated (Faithful
win) or the Traitors reach parity (Traitors win).

The natural baseline strategy for the Faithful is to eliminate a player
chosen uniformly at random. Braverman, Etesami and Mossel \cite{braverman2008}
showed that, in Mafia without detectives and with a \emph{team objective},
under which a player wins whenever their side wins, whether or not they
survive, this is optimal for the citizens, provided that the players can draw
the random choice together by acting simultaneously. Migda\l{}
\cite{migdal2010} derived an exact recurrence for the winning probability
under uniform elimination, and Wang \cite{wang2024} studied the related
Werewolf game under the same approach, with any player who departs from it
identified and eliminated in the following round.

In \emph{The Traitors} such a draw is not available. Discussion is open and
sequential, so whoever speaks last could steer any random choice made
together, and that player may be a Traitor. If the Faithful simply vote at random, an
individual random vote is indistinguishable from a strategic one, and Traitors
who concentrate their votes on one Faithful are shielded by the same noise
that protects the Faithful. The Faithful are therefore caught in a tension
between the optimality of random elimination and the vulnerability of random
voting in the game played on television.

The show already contains one source of randomness that no player controls:
the tie-break. At the final round table of the fourth UK series the vote
between a Traitor and a Faithful was tied twice, and the banishment was then
decided by a choice between two chests.\footnote{%
\url{https://en.wikipedia.org/wiki/The_Traitors_(British_TV_series)_series_4}}
Vote-Left uses this device for every banishment. In terms of the game, the
Faithful should play as follows:
\begin{enumerate}[itemsep=1pt,topsep=4pt]
    \item vote for the next surviving player in an agreed order, so that every
    player is tied on one vote;
    \item if anyone votes otherwise, vote for them on the following days, in
    turn, until they are banished;
    \item once at most \(2m_t + 2\) players remain, where \(m_t\) is the number
    of surviving Traitors, vote at random.
\end{enumerate}

\subsection*{Contributions and outline}

This paper makes four contributions. First, we introduce the Vote-Left
protocol and show that, when it is combined with a punishment rule, every
deviation in the main game \(n_t > 2m_t + 2\) that does not end the game at once
is punished, whether by a single player or by the Traitors acting together, and
that carrying out the punishment is itself a best response
(Section~\ref{sec:equilibrium}). Second, we characterise the late-game phase
\(n_t \leq 2m_t + 2\), in which punishment arrives too late, and show that the
Faithful should return to random voting there: under this switching profile no
player gains by a strategy that first departs from it in the main game
(Theorem~\ref{thm:deterrence}), and with a single Traitor it is a Perfect
Bayesian Equilibrium of the whole game (Proposition~\ref{prop:single};
Section~\ref{sec:timing}). Third, we derive exact recurrences for the Traitor
win probability under each strategy profile, together with an algorithm,
polynomial in the number of players, for the banishment probability under
random voting with collusion (Propositions~\ref{prop:migdal}
and~\ref{prop:collusion_algorithm}--\ref{prop:switch}). Fourth, we show that the
switching profile raises the Faithful's winning probability by a factor of
approximately four relative to random voting under collusion with three
Traitors, and by a factor of 62 to 84 with four (Section~\ref{sec:sim}).

With this work we release \texttt{backstab}, an open-source Python library;
all results in this paper use version 0.2.1. It provides exact rational
arithmetic for all of the recurrences and for the algorithm of
Proposition~\ref{prop:collusion_algorithm}. Releasing documented, tested
software alongside a paper is now regarded as a component of good scientific
practice \cite{taschuk2017, wilson2014}, and all numerical results in this
paper are reproducible from the library and the accompanying scripts.

A typical season starts with between 20 and 25 players, three or four of
whom are Traitors. The fan wiki of the
show\footnote{\url{https://thetraitors.fandom.com}, as downloaded on 10 November
2026.} records 92 completed seasons in 30 countries. One side won outright in 89
of them, and the Traitors won 54 of these, or 61\%
(Figure~\ref{fig:outcomes_by_country}). Under random voting, the recurrence of
Migda\l{} (Section~\ref{sec:model}) gives the Traitors a winning probability
between 0.625 and 0.864 for these starting configurations. The gap plausibly
reflects partial information from social discussion, yet the Traitor advantage
clearly persists. Vote-Left removes the vote-manipulation channel while leaving
the rest of the social dynamics of the show intact.

\begin{figure*}[htbp!]
    \begin{center}
        \includegraphics[width=0.8\textwidth]{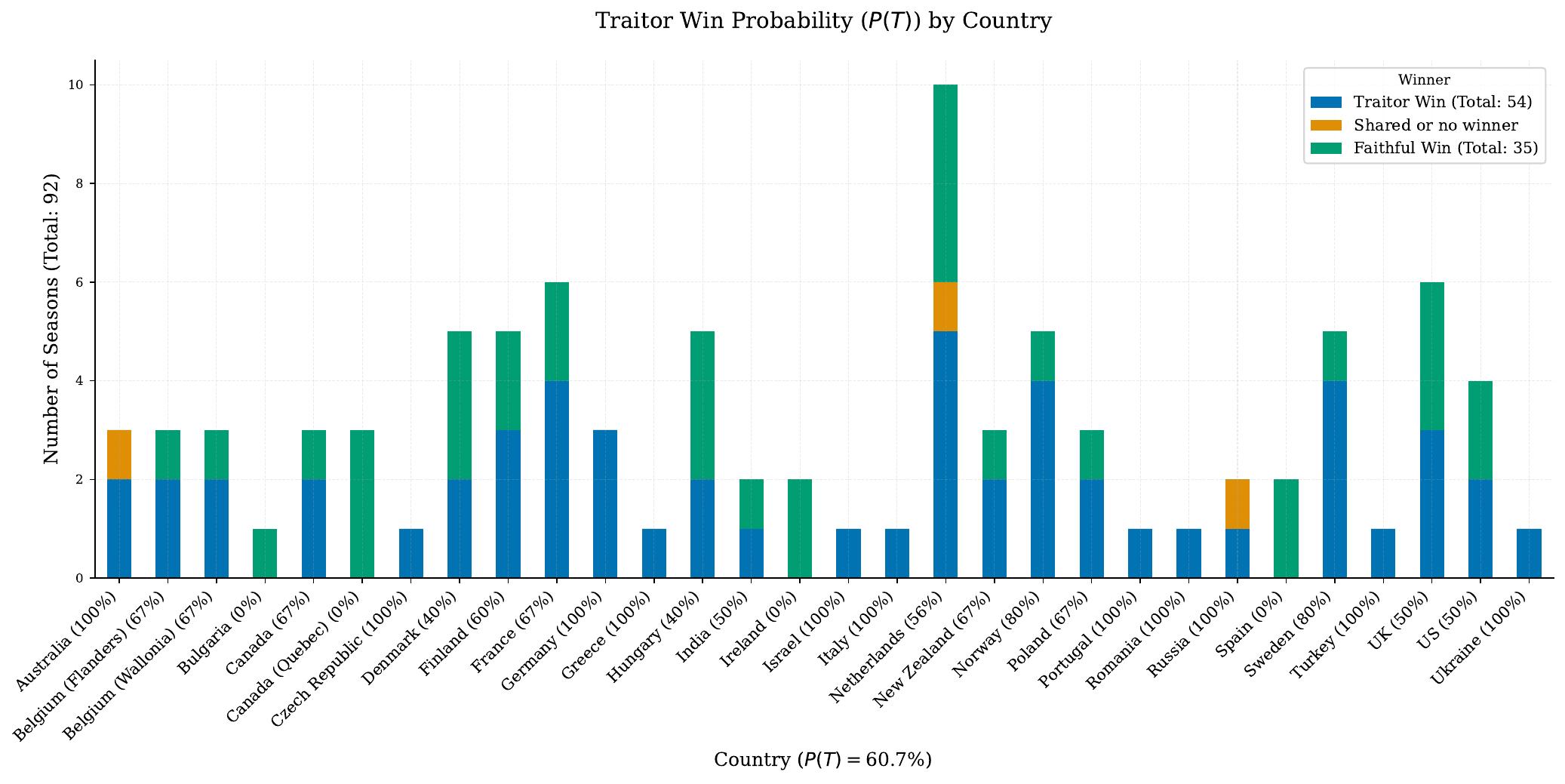}
        \caption{\textbf{Outcomes by country, 2021--2026.}
        The 92 seasons completed by November 2026 in 30 countries, with the
        editions of each country combined. The Traitors won 54 of the 89 seasons
        that one side won outright, a clear advantage to the Traitors. Two seasons
        had winners from both sides, and one (Australia, season 2) ended with no
        winner.}
        \label{fig:outcomes_by_country}
    \end{center}
\end{figure*}

\section{Model}\label{sec:model}

We define the Traitors Game \(\mathrm{TG}(n,m)\) as follows. A set of \(n\)
players is arranged in a fixed cyclic ordering. A subset of \(m\) of them are
Traitors; the remaining \(n - m\) are Faithful. The Traitors know one another,
while the Faithful know only their own role. Play alternates between a day
phase (a public vote, with the most-voted player banished, ties broken
uniformly at random, and the banished player's role revealed) and a night
phase (the Traitors murder one Faithful). Votes in a day phase are
simultaneous: each player chooses a vote without seeing the others, and all
votes are then made public. The Faithful win when all Traitors have been
banished; the Traitors win as soon as parity is reached (\(2m \geq n\)). We
write \((n_t, m_t)\) for the numbers of surviving players and Traitors at the
start of day \(t\); since every banished player's role is revealed, \(m_t\) is
public. We suppress elements of the television show such as missions,
shields, recruitment and the end game; the last two are discussed in
Section~\ref{sec:discussion}. Figure~\ref{fig:game} summarises the game
structure and contrasts random voting with Vote-Left.

\begin{figure*}[htbp!]
    \begin{center}
        \includegraphics[width=0.8\textwidth]{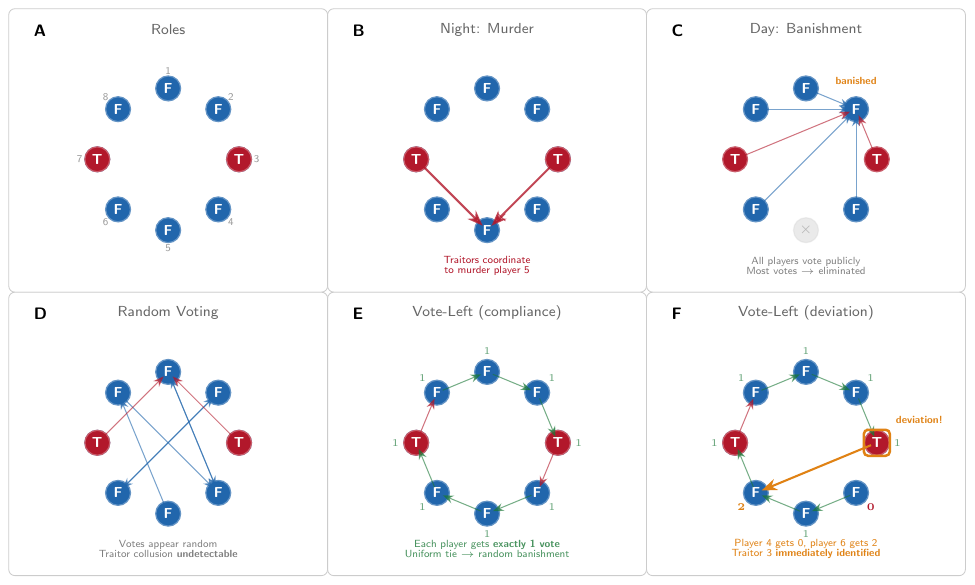}
        \caption{\textbf{The Traitors game and voting strategies.}
        \textbf{A,} Eight players (F for Faithful in blue, T for Traitor in red)
        with two Traitors. \textbf{B,} Night phase: the
        Traitors decide to murder a Faithful (player 5).
        \textbf{C,} Day phase: the public vote leads to banishment.
        \textbf{D,} Under random voting, Traitor collusion (the two
        red arrows converging on player 1) is indistinguishable from
        random noise. \textbf{E,} Under Vote-Left, each player votes
        for the next in the cycle, and every player receives exactly
        one vote. \textbf{F,} A Traitor who deviates from Vote-Left
        breaks the uniform count (player 4 receives no votes and player 6
        receives two), making the deviation immediately identifiable.}
        \label{fig:game}
    \end{center}
\end{figure*}

\begin{assumption}[Individual maximisation]\label{ass:individual}
All players maximise their individual winning probability. A player wins only
if their side wins and they are still in the game when it is decided.
\end{assumption}

This matches the show, where only the players who reach the end share the
prize.

\begin{remark}[Win condition and tie-breaking]\label{rem:boundary}
We use the parity win condition (\(2m \geq n\)), appropriate for \emph{The
Traitors}' rules. Migda\l{} \cite{migdal2010} uses the strict-majority
condition (\(m > n - m\)); the two agree for odd \(n\) and differ only at exact
ties for even \(n\). We assume uniform random tie-breaking. In practice the
show uses a revote followed by a random draw, and in the versions we have
checked the tied players do not vote in the revote. Under Vote-Left every
player is tied, so no player is eligible to vote again and the banishment falls
to the chance device. Our results need only that it selects uniformly among
the tied players.
\end{remark}

A \emph{strategy} for player \(i\) maps observable game history to a
distribution over surviving players (excluding themself). We write \(L(i)\)
for the next surviving player after~\(i\) in the cyclic ordering.

\subsection{Baseline winning probability}

The following recurrence, due to Migda\l{} \cite{migdal2010}, gives the exact
Traitor win probability under mutual random play. We restate it here using the
parity win condition of Remark~\ref{rem:boundary}, and provide a short proof
for the reader's convenience.

\begin{proposition}[Migda\l{} recurrence, adapted]\label{prop:migdal}
Let \(w(n, m)\) denote the Traitor win probability under mutual random play
from state \((n, m)\). Then
\begin{equation}\label{eq:migdal}
w(n,m) = \begin{cases}
  0 & \text{if } m = 0, \\
  1 & \text{if } n \leq 2m, \\[4pt]
  \dfrac{n-m}{n}\,w(n-2,m) + \dfrac{m}{n}\,w(n-2,m-1) & \text{if } n > 2m.
\end{cases}
\end{equation}
\end{proposition}

\begin{proof}[Proof (following Migda\l{}~\cite{migdal2010})]
The boundary cases follow from the win conditions. For \(n > 2m\), a round is
a day phase followed by a night phase. In the day phase each surviving player
is equally likely to be banished, so with probability \(m/n\) a Traitor is
banished and the state becomes \((n - 1, m - 1)\)
(Figure~\ref{fig:migdal_recursion}C), and with probability \((n - m)/n\) a
Faithful is banished and the state becomes \((n - 1, m)\)
(Figure~\ref{fig:migdal_recursion}B). In the night phase the Traitors murder a
Faithful, which removes a further player and leaves \(m\) unchanged, so after
the round the state is \((n - 2, m - 1)\) (Figure~\ref{fig:migdal_recursion}E)
or \((n - 2, m)\) (Figure~\ref{fig:migdal_recursion}D) respectively.
Conditioning on the outcome of the day phase gives~\eqref{eq:migdal}. The
recurrence terminates since \(n\) falls by two each round.
\end{proof}

\begin{figure}[tbp]
\centering
\includegraphics[width=\textwidth]{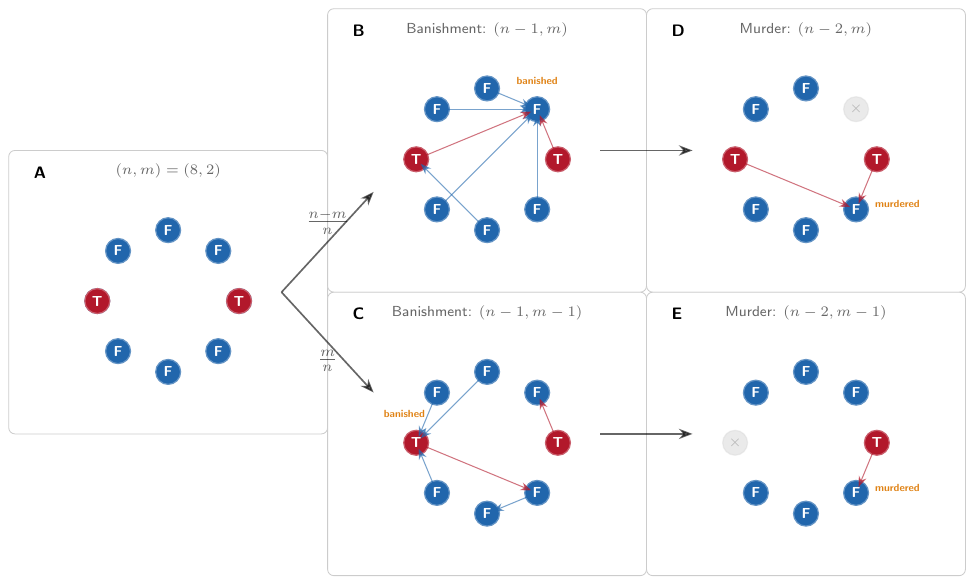}
\caption{\textbf{Recursion relationship for \(w(n, m)\).}
\textbf{A,} Initial state for \(w(8, 2)\). Under random voting the day phase
banishes a Faithful with probability \((n - m)/n\) (\textbf{B}) or a Traitor
with probability \(m/n\) (\textbf{C}). The night murder then leads to
\textbf{D} or \textbf{E} respectively.}
\label{fig:migdal_recursion}
\end{figure}

We now define the cyclic voting protocol referred to as Vote-Left.

\begin{definition}[Cyclic voting protocol]
A \emph{cyclic voting protocol} prescribes that each player votes for a target
determined by a shared, fixed single-cycle permutation of the surviving
players. The \emph{Vote-Left protocol} is the instance where each player votes
for \(L(i)\).
\end{definition}

The ordering need not be the seating at the round table, which changes during
the show. Any public ordering that the players agree in advance will do, such
as the alphabetical order of their first names.

\begin{observation}[Uniform marginal and observability]\label{obs:main}
Under full compliance: (i)~every player receives exactly one vote, producing a
uniform \(1/n_t\) banishment probability, identical to random voting; (ii)~any
deviation is immediately and publicly identifiable, since prescribed votes are
deterministic functions of public information.
\end{observation}

Random voting achieves (i) but not (ii): a Traitor who concentrates their
vote is indistinguishable from a Faithful who happened to draw that target at
random. Vote-Left achieves both properties simultaneously.

\section{Deviation and punishment}\label{sec:equilibrium}

\subsection{Strategy profile}

\begin{definition}[Vote-Left with Punishment]\label{def:equil}
Vote-Left with Punishment is the strategy profile \(\sigma^*\), governed by a
public state \(s \in \{\mathtt{comply},\, \mathtt{punish}(j)\}\) common to all
players:
\begin{itemize}[leftmargin=1.4em,itemsep=1pt,topsep=2pt]
\item In state \(\mathtt{comply}\), every surviving player \(i\) votes for
  \(L(i)\).
\item In state \(\mathtt{punish}(j)\), every surviving player other than \(j\)
  votes for \(j\), and \(j\) votes for \(L(j)\).
\item A player whose vote differs from the one prescribed joins a public queue
  of deviators; when several join on the same day, they do so in the fixed
  cyclic ordering. While the queue is non-empty the state is
  \(\mathtt{punish}(j)\), with \(j\) the first player in it. A player leaves
  the queue when they are banished or murdered, and when the queue is empty
  the state returns to \(\mathtt{comply}\).
\end{itemize}
At night the Traitors murder a uniformly random Faithful, other than a player
who is due to be punished the next day.
\end{definition}

The exception in the murder rule reflects the Traitors' interest: a Faithful
player who is due to be punished will be banished with certainty, and
murdering them would cancel that banishment.

\subsection{Deviation payoff analysis}

We consider a player \(i\) who departs from \(\sigma^*\) in state
\(\mathtt{comply}\) while every other player follows it. The key distinction
is whether the game is decided before the punishment round. We call a
deviation \emph{terminal} if, with the votes it produces, the game is decided
during that day phase or the night that follows, and \emph{non-terminal}
otherwise. No threat of future punishment can deter a terminal deviation.

\begin{lemma}[Certain punishment]\label{lemma:punish}
Let \((n, m)\) be a state with \(n > 2m + 2\), and suppose player \(i\)
deviates from \(\sigma^*\) while all other players follow it. If the game is
not decided before the next day phase, then \(i\) does not survive that day
phase.
\end{lemma}

\begin{proof}
The deviation is publicly observable, so the state becomes
\(\mathtt{punish}(i)\). If \(i\) survives the night there are
\(n - 2 > 2m \geq 2\) players. Every player other than \(i\) votes for \(i\),
so \(i\) receives \(n - 3 \geq 2\) votes, whereas any other player receives at
most the single vote that \(i\) casts. Hence \(i\) is banished.
\end{proof}

The punishment does not rely on the Traitors taking part in it. After the
night there are \(n - 2 - m > m\) Faithful, so a deviating Traitor receives
more votes than any other player even if the other Traitors refuse to punish.
The same counting covers joint deviations by the Traitors, who know one
another and can coordinate.

\begin{proposition}[Joint deviations by the Traitors]\label{prop:coalition}
Let \((n, m)\) be a state with \(n > 2m + 2\), and suppose that on a day in
state \((n, m)\) a set of Traitors departs from \(\sigma^*\) while the Faithful
follow it. Then every Traitor whose vote that day differs from the prescribed
one is banished before the Traitors reach parity, whatever the Traitors do
afterwards, and so wins with probability zero.
\end{proposition}

\begin{proof}
Write \(d = n - 2m\) for the difference between the numbers of Faithful and
Traitors, so that \(d > 2\) on the day of the deviation. That day and night
remove at most two Faithful, so afterwards \(d \geq 1\). On each later day
with a deviator awaiting punishment, every Faithful votes for the deviator,
who therefore receives at least \(n - m\) votes, whereas any other player
receives at most the \(m\) votes of the Traitors. Since \(d \geq 1\) the
deviator is banished, and the night then removes a Faithful, so \(d\) is
unchanged. Hence \(d \geq 1\) until every deviator of the first day has been
banished, and the Traitors cannot reach parity before then.
\end{proof}

Only a Traitor whose own vote differs from their \(L(i)\) is identified.
Traitors who wish to concentrate their votes on one Faithful therefore do best
to choose the \(L(i)\) of one of them, so that the deviation costs a single
Traitor. Proposition~\ref{prop:coalition} concerns individual payoffs; we
return to the team objective in Section~\ref{sec:discussion}.

\begin{lemma}[Terminal deviations]\label{lemma:terminal}
Let \((n, m)\) be a state with \(n > 2m + 2\), and suppose player \(i\) votes
for \(j \neq L(i)\) while all other players follow \(\sigma^*\). Then \(j\) is
banished with probability one. The deviation is terminal if and only if
\(m = 1\) and \(j\) is the Traitor, which requires \(i\) to be Faithful.
\end{lemma}

\begin{proof}
Redirecting \(i\)'s vote leaves \(L(i)\) with none, gives \(j\) two, and
leaves every other player with one, so \(j\) is banished. If \(j\) is Faithful,
or a Traitor with \(m \geq 2\), at least one Traitor remains, and the
\(n - 2 > 2m\) players left after the night do not give the Traitors parity,
so the game is not decided. If \(j\) is the only Traitor the Faithful win at
once, and since a Traitor cannot vote for themself, \(i\) is Faithful.
\end{proof}

A terminal deviation is a gamble on the role of its target, so its value
depends on what the deviator believes.

\begin{remark}[Beliefs]\label{rem:concept}
The roles of the surviving players stay hidden until the game is decided, and
the appropriate solution concept is the Perfect Bayesian Equilibrium
\cite{fudenberg1991}: at every information set each player's strategy is a
best response given their beliefs, and beliefs follow from Bayes' rule
wherever it applies. When every player follows Vote-Left the votes carry no
information about roles and the murders are uniform, so a Faithful player
holds a uniform belief over the other surviving players. After a deviation
Bayes' rule is silent, since a deviation has probability zero, and the belief
is part of the model. We take every Faithful player to believe that a deviator
\(j\) is a Traitor until \(j\)'s role is revealed, and to return to a uniform
belief if \(j\) turns out to be Faithful. One motivation is that players might
occasionally depart from the protocol by mistake, with Traitors far more
likely to do so than the Faithful.
\end{remark}

\begin{proposition}[Credible punishment]\label{prop:credible}
Consider a day phase in state \(\mathtt{punish}(j)\) with \(n\) surviving
players, and suppose that the Faithful hold the beliefs of
Remark~\ref{rem:concept}. If \(n \geq 5\), no single player's vote changes the
banishment, and no player gains by voting for anyone other than \(j\). If
\(n \leq 4\), a single Traitor remains, and voting for \(j\) is a strict best
response for every player other than \(j\).
\end{proposition}

\begin{proof}
Under the prescribed votes \(j\) receives \(n - 1\) votes and every other
player at most one. A single player who votes elsewhere leaves \(j\) with
\(n - 2\) votes and gives any other player at most two. If \(n \geq 5\) then
\(n - 2 \geq 3\), so \(j\) is still banished. If the profile punishes the
departure, the deviator joins the queue and play is unchanged until the night
before their punishment falls due; on that day, if the game is still open,
every other player votes for them and, since the Faithful outnumber the
Traitors, they are banished. If the profile does not punish it, which under
the switching profile of Section~\ref{sec:switch} happens in the late game,
play continues exactly as under compliance. In neither case can the departure
raise the player's winning probability. If \(n \leq 4\), the Traitors do not
have parity, so \(m = 1\). A Faithful player who believes that \(j\) is the
Traitor wins with certainty by banishing \(j\), and any other vote banishes
\(j\) with probability less than one. If \(j\) is Faithful, the Traitor
reaches parity by banishing \(j\) and does so with certainty by voting for
\(j\).
\end{proof}

In summary, while \(n > 2m + 2\) a non-terminal deviation, by a single player
or by several Traitors, leaves every deviator with a winning probability of
zero, and carrying out the punishment is a best response. The only terminal
deviation is a Faithful player redirecting their vote when a single Traitor
remains (Lemma~\ref{lemma:terminal}). Whether it pays depends on what
compliance is worth, and so on how the game is played once \(n \leq 2m + 2\),
which we now consider.

\section{Traitor deviation timing}\label{sec:timing}

Section~\ref{sec:equilibrium} establishes that punishment follows any
non-terminal deviation in a state satisfying \(n_t > 2m_t + 2\), the
\emph{main game}. In the complementary \emph{late-game phase},
\(n_t \leq 2m_t + 2\), punishment arrives too late to deter the Traitors.
When \(2m_t < n_t \leq 2m_t + 2\), suppose the Traitors collude by all voting
for the same Faithful target while the Faithful follow Vote-Left. With
\(m_t \geq 2\) the target receives at least two votes and every other player
at most one; with \(m_t = 1\) the same holds provided the target is not the
Traitor's own \(L(i)\). The target is therefore banished. After the following
night murder there are \(n_t - 2 \leq 2m_t\) players and still \(m_t\)
Traitors, so the Traitors reach parity before the punishment round can take
place (Figure~\ref{fig:late_game}).

\begin{figure}[tbp]
\centering
\includegraphics[width=0.85\textwidth]{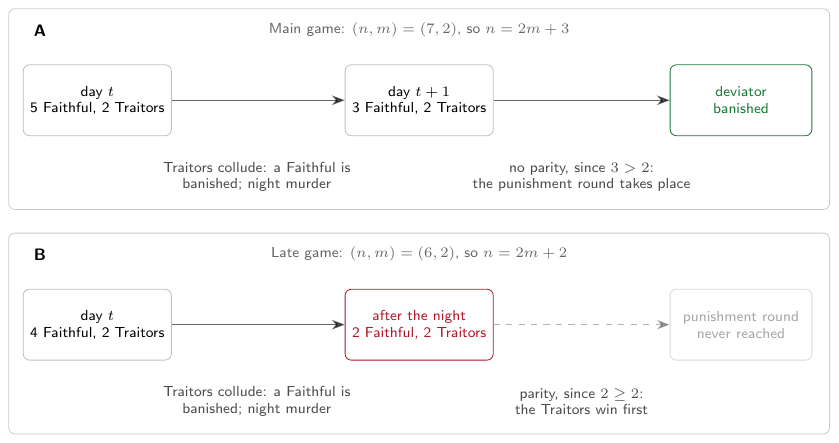}
\caption{\textbf{Why punishment arrives too late.} A round removes two
players, so a punishment decided on day \(t\) is carried out on day \(t + 1\)
with two fewer players. \textbf{A,} From \((7, 2)\), in the main game, the
Traitors collude and a Faithful is banished; after the night murder the
Faithful still outnumber the Traitors, the punishment round takes place, and
the deviator is banished. \textbf{B,} From \((6, 2)\), in the late game, the
same collusion and night murder leave two Faithful and two Traitors. The
Traitors reach parity and win before the punishment round.}
\label{fig:late_game}
\end{figure}

We therefore define \(\sigma^\dagger\) as the strategy profile in which the
Faithful play Vote-Left with Punishment (Definition~\ref{def:equil}) while the
Traitors comply with Vote-Left whenever \(n_t > 2m_t + 2\) and collude when
\(n_t \leq 2m_t + 2\). For \(m_t \in \{2, 3, 4\}\) the late-game threshold is
\(n_t \leq 6\), \(8\) or \(10\) respectively. Since a round removes two
players, the banishment of a Traitor leaves the game in the late-game phase,
so a late game entered with \(m_t\) Traitors can last up to \(m_t\) rounds.
Under \(\sigma^\dagger\) the Faithful lose every late-game state with
certainty.

\subsection{Collusion banishment probability}
\label{sec:collusion-banishment}

We define \(\sigma^\ddagger\) as the strategy profile in which the Traitors
always collude by voting for the same Faithful target, while the Faithful vote
uniformly at random. This profile serves as the collusion baseline: it captures
the advantage that the Traitors can extract under random voting.

The recurrence for the Traitor win probability under \(\sigma^\ddagger\)
requires the probability that the day-phase banishment lands on a Faithful.
Label the players \(0, \dots, n - 1\), with the Faithful
\(F = \{0, \dots, n - m - 1\}\) and \(T^* = 0\) the colluding target, and let
\(\Omega\) be the set of Faithful voting profiles \(v = (v_i)_{i \in F}\) with
\(v_i \in \{0, \dots, n - 1\}\) and \(v_i \neq i\), each of which is equally
likely. Write
\(C_j(v) = m \mathbf{1}_{\{j = T^*\}} + \sum_{i \in F} \mathbf{1}_{\{v_i = j\}}\)
for the votes received by player \(j\), \(M(v) = \max_j C_j(v)\), and
\(W(v) = \{j : C_j(v) = M(v)\}\) for the tied players. With ties broken
uniformly at random, the probability that a Faithful is banished is
\begin{equation}\label{eq:coll_banish}
p_F^{\ddagger}(n,m)
= \frac{1}{|\Omega|} \sum_{v\in\Omega} \frac{|W(v)\cap F|}{|W(v)|}.
\end{equation}
Evaluating~\eqref{eq:coll_banish} directly requires
\(|\Omega| = (n - 1)^{n - m}\) terms, about \(10^{25}\) at \((22, 3)\).

\begin{proposition}[Exact computation of \(p_F^\ddagger\)]
\label{prop:collusion_algorithm}
The probability \(p_F^\ddagger(n, m)\) can be computed in exact rational
arithmetic with \(O(m n^5)\) arithmetic operations.
\end{proposition}

The proof, by inclusion and exclusion over self-votes and generating functions
grouped by the maximum vote count, is in Appendix~\ref{app:algorithm}. Our
implementation in \texttt{backstab} computes \(p_F^\ddagger(25, 3)\) in about
a tenth of a second, and agrees with direct enumeration
of~\eqref{eq:coll_banish} for every \(n \leq 8\).

\subsection{Recurrences for
\texorpdfstring{\(\sigma^\ddagger\)}{sigma-ddagger} and
\texorpdfstring{\(\sigma^\dagger\)}{sigma-dagger}}
\label{sec:recurrences}

Replacing the uniform banishment probability \((n - m)/n\) of the Migda\l{}
recurrence with \(p_F^\ddagger\) characterises the Traitor win probability
under \(\sigma^\ddagger\), and changing its boundary characterises it under
\(\sigma^\dagger\).

\begin{proposition}[Traitor win probability under \(\sigma^\ddagger\)]
\label{prop:rvc}
Let \(w_{\ddagger}(n,m)\) denote the Traitor win probability under
\(\sigma^\ddagger\) from state \((n, m)\). Then
\begin{equation}\label{eq:rvc}
w_{\ddagger}(n,m) =
\begin{cases}
0 & m = 0, \\[4pt]
1 & n \leq 2m, \\[4pt]
p_F^{\ddagger}(n,m)\, w_{\ddagger}(n-2,m)
+ \bigl(1 - p_F^{\ddagger}(n,m)\bigr)\, w_{\ddagger}(n-2,m-1)
& n > 2m.
\end{cases}
\end{equation}
\end{proposition}

\begin{proposition}[Traitor win probability under \(\sigma^\dagger\)]
\label{prop:vlopt}
Let \(w_{\dagger}(n, m)\) denote the Traitor win probability under
\(\sigma^\dagger\) from state \((n, m)\). Then
\begin{equation}\label{eq:vlopt}
w_{\dagger}(n, m) =
\begin{cases}
0 & m = 0, \\[4pt]
1 & n \leq 2m + 2, \\[4pt]
\dfrac{n - m}{n}\, w_{\dagger}(n - 2,\, m)
 + \dfrac{m}{n}\, w_{\dagger}(n - 2,\, m - 1)
& n > 2m + 2.
\end{cases}
\end{equation}
\end{proposition}

\begin{proof}[Proof of Propositions~\ref{prop:rvc} and~\ref{prop:vlopt}]
Each round is a day-phase banishment followed by a night-phase murder, as in
Proposition~\ref{prop:migdal}. Under \(\sigma^\ddagger\) the banishment is a
Faithful with probability \(p_F^\ddagger(n, m)\). Under \(\sigma^\dagger\),
for \(n > 2m + 2\) every player follows Vote-Left, so the banishment is
uniform, and for \(2m < n \leq 2m + 2\) the collusion described above forces
the banishment of a Faithful and the Traitors win.
\end{proof}

\subsection{Switching to random voting}\label{sec:switch}

Under \(\sigma^\dagger\) the late game is lost for the Faithful, because their
Vote-Left votes give the colluding Traitors' target the second vote that it
needs. Random voting does not have this weakness: the colluding Traitors still
concentrate their votes, but the Faithful votes are spread at random and can
outweigh them. In the late game punishment can no longer deter collusion, so
the Faithful lose nothing by giving up the observability that Vote-Left
provides.

\begin{definition}[Vote-Left then random voting]\label{def:switch}
The switching profile \(\sigma^s\) is governed by the same public state and
queue of deviators as \(\sigma^*\) (Definition~\ref{def:equil}), and the
Traitors murder as in \(\sigma^*\). In state \(\mathtt{comply}\), every
surviving player votes for \(L(i)\) while \(n_t > 2m_t + 2\); once
\(n_t \leq 2m_t + 2\), the Faithful vote uniformly at random among the other
surviving players, and the Traitors collude on a single Faithful target drawn
uniformly at random. A player who deviates on a day with
\(n_t > 2m_t + 2\) joins the queue as before, and the punishment of every
player in the queue is carried out even if the late game has been reached by
then. Deviations on a day with \(n_t \leq 2m_t + 2\) are not punished.
\end{definition}

Figure~\ref{fig:state_space} shows where the switch falls in the space of
states. The late game is a band two players wide, and a round moves the state
either two places to the left or, when a Traitor is banished, two places to
the left and one down, parallel to the band. A game that enters the late game
therefore stays in it, and the Faithful switch once. A game that starts at a
televised configuration plays at least six rounds in the main game before it
can reach the late game. The switch requires no extra information, since
\(m_t\) is the initial number of Traitors less the number of banished players
revealed to be Traitors. In the main game \(\sigma^s\) coincides with
\(\sigma^*\) and completes every punishment that has started, so
Lemmas~\ref{lemma:punish} and~\ref{lemma:terminal} and
Proposition~\ref{prop:coalition} apply to it unchanged.

\begin{figure}[tbp]
\centering
\includegraphics[width=0.9\textwidth]{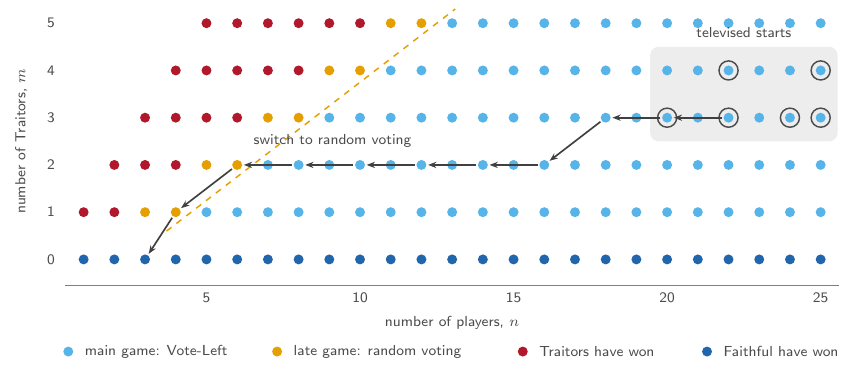}
\caption{\textbf{The main game and the late game.} Each dot is a state
\((n, m)\), coloured light blue in the main game \(n > 2m + 2\), where the
switching profile \(\sigma^s\) prescribes Vote-Left, orange in the late game
\(2m < n \leq 2m + 2\), where it prescribes random voting, red where the
Traitors have won, and dark blue where the Faithful have won. The dashed line
is the boundary between the main game and the late game. The shaded box is the
televised range, 20 to 25 players of whom three or four are Traitors, and
circles mark the six televised starting configurations. The arrows show one
possible game from \((22, 3)\): each round removes two players, and a banished
Traitor also lowers \(m\) by one. The game enters the late game at \((6, 2)\)
and stays in it until the last Traitor is banished, leaving three players and
a win for the Faithful.}
\label{fig:state_space}
\end{figure}

\begin{proposition}[Traitor win probability under \(\sigma^s\)]
\label{prop:switch}
Let \(w_s(n, m)\) denote the Traitor win probability under \(\sigma^s\) from
state \((n, m)\). Then
\begin{equation}\label{eq:switch}
w_s(n, m) =
\begin{cases}
0 & m = 0, \\[4pt]
w_{\ddagger}(n, m) & n \leq 2m + 2, \\[4pt]
\dfrac{n - m}{n}\, w_s(n - 2,\, m)
 + \dfrac{m}{n}\, w_s(n - 2,\, m - 1)
& n > 2m + 2.
\end{cases}
\end{equation}
\end{proposition}

\begin{proof}
For \(n > 2m + 2\) every player follows Vote-Left, so the banishment is
uniform. Once \(n \leq 2m + 2\), every later state is also in the late game,
since a round moves \((n, m)\) to \((n - 2, m)\) or \((n - 2, m - 1)\) and
\(n - 2 \leq 2(m - 1) + 2\) (Figure~\ref{fig:state_space}). Play from that
point on is that of \(\sigma^\ddagger\), and Proposition~\ref{prop:rvc}
applies.
\end{proof}

Appendix~\ref{app:validation} checks Propositions~\ref{prop:migdal}
and~\ref{prop:rvc}--\ref{prop:switch}, and their implementation in
\texttt{backstab}, against direct simulation of the game.

With a single Traitor, collusion has no force, since a coalition of one casts
the same single vote as any other player. Hence \(w_s(n, 1) = w(n, 1)\) for
every \(n\), whereas \(w_\dagger(n, 1) = 1\) for \(n \leq 4\).

We can now value compliance for the remaining deviation. Let \(h(n)\) denote
the probability that a given Faithful player wins under \(\sigma^s\) from the
start of a day phase in state \((n, 1)\). In the late game, \(n \leq 4\), every
player votes uniformly at random for one of the others, so the Traitor is
banished with probability \(1/n\), and any other banishment hands the Traitor
parity. In the main game the Traitor is again banished with probability
\(1/n\); with probability \((n - 2)/n\) another Faithful is banished, and the
given player survives the night with probability \((n - 3)/(n - 2)\). Hence
\[
h(n) = \begin{cases}
  1/3 & \text{if } n = 3, \\[4pt]
  1/4 & \text{if } n = 4, \\[4pt]
  \dfrac{1}{n} + \dfrac{n - 3}{n}\, h(n - 2) & \text{if } n > 4.
\end{cases}
\]

\begin{theorem}[Deterrence in the main game]\label{thm:deterrence}
Let \((n, m)\) be a state with \(n > 2m + 2\) in public state
\(\mathtt{comply}\), at which the Faithful hold the uniform beliefs of
Remark~\ref{rem:concept}. Then no player can increase their winning
probability with a strategy that first departs from \(\sigma^s\) in state
\((n, m)\).
\end{theorem}

\begin{proof}
Which Faithful the Traitors murder does not affect the Traitors' winning
probability, since banishments are uniform in the main game and symmetric in
the late game, so a departure matters only through a vote that differs from
\(L(i)\). We split on the deviating player.

\emph{Case 1 (Traitor, or Faithful with \(m \geq 2\)).} By
Lemma~\ref{lemma:terminal} the deviation is non-terminal, so by
Lemma~\ref{lemma:punish} the deviator wins with probability zero. Under
compliance they win with positive probability.

\emph{Case 2 (Faithful, \(m = 1\)).} By Lemma~\ref{lemma:terminal} the
deviation is terminal exactly when its target is the Traitor, which has
probability \(1/(n - 1)\), and otherwise it yields zero. Under compliance the
player wins with probability \(h(n)\). Direct computation gives
\(h(5) = 1/3 > 1/4\) and \(h(6) = 7/24 > 1/5\). For \(n \geq 7\), if
\(h(n - 2) > 1/(n - 3)\), then
\(h(n) > 1/n + 1/n = 2/n > 1/(n - 1)\), so by induction
\(h(n) > 1/(n - 1)\) for every \(n \geq 5\).
\end{proof}

The switch to random voting is what makes Theorem~\ref{thm:deterrence} hold.
Under \(\sigma^\dagger\) the late game is worth nothing to a Faithful player,
so the same recurrence starts from zero at \(n \leq 4\) and gives \(1/5\) and
\(1/6\) at \((5, 1)\) and \((6, 1)\), less than the \(1/4\) and \(1/5\) of
redirecting a vote. A Faithful player would gamble on ending the game at once,
and one of these two states is reached with positive probability from every
televised configuration. Neither \(\sigma^*\) nor \(\sigma^\dagger\) is
therefore an equilibrium of the full game: under \(\sigma^*\) a Traitor gains
by colluding in the late game, and under \(\sigma^\dagger\) a Faithful player
gains by this gamble.

With a single Traitor the late game lasts a single day phase.

\begin{proposition}[Equilibrium with a single Traitor]\label{prop:single}
For every \(n\), the profile \(\sigma^s\) in \(\mathrm{TG}(n, 1)\), together
with the beliefs of Remark~\ref{rem:concept}, is a weak Perfect Bayesian
Equilibrium.
\end{proposition}

\begin{proof}
By the one-shot deviation principle for finite games \cite{fudenberg1991} it
suffices to check single departures at each information set. In state
\(\mathtt{comply}\) with \(n \geq 5\) this is Theorem~\ref{thm:deterrence},
whose argument uses only the beliefs, and in state \(\mathtt{punish}(j)\) it
is Proposition~\ref{prop:credible}. At \((3, 1)\) and \((4, 1)\) in state
\(\mathtt{comply}\) the game ends within one round, and every player's
opponents treat the other players symmetrically, so every vote gives the same
winning probability. At night, only the murder of a Faithful \(j\) due for
punishment matters. Sparing \(j\) leads to \((n - 4, 1)\) and murdering \(j\)
to \((n - 2, 1)\), where \(n\) is the number of players on the day before
the punishment falls due. The Faithful win probability \(f(k) = 1 - w(k, 1)\) satisfies
\(f(k) - f(k - 2) = (1 - f(k - 2))/k > 0\) by~\eqref{eq:migdal}, so
\(w(n - 4, 1) \geq w(n - 2, 1)\) and the Traitor prefers to spare \(j\).
\end{proof}

With two or more Traitors the late game can last several rounds. For the
Traitors, joining the collusion there is a best response: we verify this by
exact enumeration of the Faithful votes in every late-game state with
\(2 \leq m \leq 4\), against every alternative vote, and the smallest margin
is 0.0065, at \((9, 4)\). For the Faithful the question remains open, since
votes cast in one late-game round carry information about who the Traitors
are, and we return to it in
Section~\ref{sec:discussion}.

\section{Numerical evaluation}\label{sec:sim}

In this section we use the exact recurrences of
Sections~\ref{sec:model} and~\ref{sec:timing}, equations
(\ref{eq:migdal}--\ref{eq:switch}), to quantify how Vote-Left changes the
Faithful's winning probability. Every value in this section is
exact.

Figure~\ref{fig:exact} shows the Faithful win probability under each profile
for \(m \in \{2, 3, 4, 5\}\) and \(n \leq 25\). Collusion under random voting
(\(\sigma^\ddagger\)) is severe, and with four or more Traitors it leaves the
Faithful with virtually no chance. Vote-Left with late-game collusion
(\(\sigma^\dagger\)) improves on it when \(m \geq 3\) in every game that starts
in the main game, but with \(m = 2\) it is worse for the Faithful: once one
Traitor is banished the game is often already in a single-Traitor late game,
which \(\sigma^\dagger\) concedes. The switching
profile \(\sigma^s\) lies below random voting without collusion throughout,
and close to it at the televised sizes.

\begin{figure}[tbp]
\centering
\includegraphics[width=0.9\textwidth]{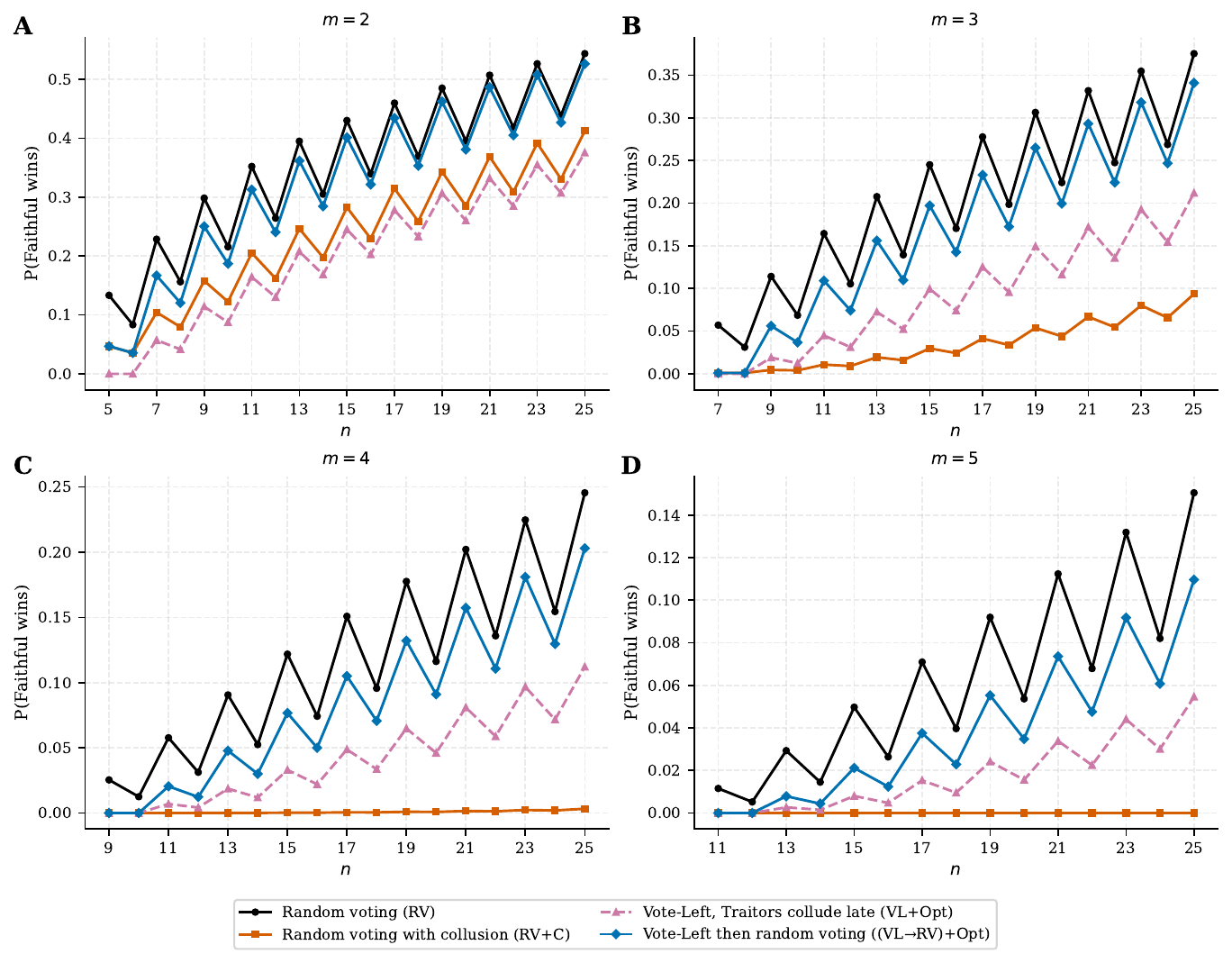}
\caption{\textbf{Exact Faithful win probability by strategy profile.} Each
panel shows \(1 - w\) as a function of \(n\) for one number of Traitors:
\textbf{A,} \(m = 2\); \textbf{B,} \(m = 3\); \textbf{C,} \(m = 4\);
\textbf{D,} \(m = 5\). The profiles are random voting (RV,
Proposition~\ref{prop:migdal}), random voting with collusion
(\(\sigma^\ddagger\), RV+C, Proposition~\ref{prop:rvc}), Vote-Left with
late-game collusion (\(\sigma^\dagger\), VL+Opt,
Proposition~\ref{prop:vlopt}), and the switching profile (\(\sigma^s\),
(VL\(\to\)RV)+Opt, Proposition~\ref{prop:switch}). The switching profile lies
below random voting and well above both profiles in which the Traitors'
collusion goes unchecked.}
\label{fig:exact}
\end{figure}

Table~\ref{tab:paramsweep} reports the values at 13 configurations: six small
ones, \((15, 2)\), and the six typical televised configurations. The
switching profile gives the lowest Traitor win probability of the three
profiles in which the Traitors collude, in every configuration that starts in
the main game; at \((11, 5)\) the game starts in the late game, where
\(\sigma^s\) and \(\sigma^\ddagger\) coincide. Relative to \(\sigma^\ddagger\),
the switching profile raises the Faithful's winning probability by a factor of
1.4 at \((15, 2)\), between 3.6 and 4.5 at the televised configurations with
three Traitors, and by factors of 84 and 62 at \((22, 4)\) and \((25, 4)\).
Relative to \(\sigma^\dagger\) it raises it by a factor of between 1.6 and 1.9
at the televised configurations. It lies between 2.2 and 4.3 percentage points
above the random voting value \(w\), which is the value the Faithful would
obtain if the Traitors never colluded.

\begin{table}[htbp!]
\centering
\caption{Exact Traitor win probability under random voting (\(w\),
Proposition~\ref{prop:migdal}), random voting with collusion (\(w_\ddagger\),
Proposition~\ref{prop:rvc}), Vote-Left with late-game collusion (\(w_\dagger\),
Proposition~\ref{prop:vlopt}) and the switching profile (\(w_s\),
Proposition~\ref{prop:switch}); lower values are better for the Faithful.}
\label{tab:paramsweep}
\smallskip
\begin{tabular}{@{}ccccc@{}}
\toprule
\((n, m)\) & \(w\) (RV) & \(w_{\ddagger}\) (RV+C) & \(w_{\dagger}\) (VL+Opt) & \(w_{s}\) ((VL\(\to\)RV)+Opt) \\
\midrule
(7,\,2) & 0.771 & 0.896 & 0.943 & 0.833 \\
(8,\,2) & 0.844 & 0.920 & 0.958 & 0.880 \\
(9,\,3) & 0.886 & 0.995 & 0.981 & 0.944 \\
(10,\,3) & 0.931 & 0.996 & 0.988 & 0.963 \\
(11,\,3) & 0.835 & 0.989 & 0.955 & 0.891 \\
(11,\,5) & 0.988 & 1.000 & 1.000 & 1.000 \\
\midrule
(15,\,2) & 0.570 & 0.717 & 0.755 & 0.599 \\
(20,\,3) & 0.776 & 0.956 & 0.884 & 0.800 \\
(22,\,3) & 0.752 & 0.945 & 0.864 & 0.776 \\
(24,\,3) & 0.731 & 0.934 & 0.846 & 0.753 \\
(25,\,3) & 0.625 & 0.906 & 0.788 & 0.659 \\
(22,\,4) & 0.864 & 0.999 & 0.941 & 0.889 \\
(25,\,4) & 0.754 & 0.997 & 0.888 & 0.797 \\
\bottomrule
\end{tabular}

\end{table}

\section{Discussion}\label{sec:discussion}

Our analysis covers the core Traitors game: a fixed roster of players, a
sequence of day votes and night murders, and a parity win condition for the
Traitors. The most pressing open question is an equilibrium of the full game
with two or more Traitors. Under the switching profile the main game and its
punishments are self-enforcing, the whole profile is an equilibrium with a
single Traitor, and the Traitors' late-game collusion is a best response in the
televised range. What remains is the Faithful's play in a late game that lasts
several rounds, where votes cast in one round carry information about who the
Traitors are. Random voting is not always the best the Faithful can do as a
group there: if they all vote for the first surviving player in the agreed
order, a Traitor is banished whenever that player is a Traitor. The Traitors
can plan their main-game murders around such an order, however, and against
Traitors who do so this concentrated vote is worse for the Faithful than
\(\sigma^s\) at every televised configuration. Whether some late-game protocol
improves on \(\sigma^s\) remains open.

Several features of the game lie outside our model. Our results rely on
individual payoffs (Assumption~\ref{ass:individual}); under the team objective
of Braverman, Etesami and Mossel a player may accept their own punishment for
the benefit of their side. We compare such a sacrifice with compliance exactly
in every main-game state with \(n \leq 25\) and \(m \leq 5\): it never pays for
the Faithful, but it pays for the Traitors in eight states close to the late
game, with the largest gain, 1.9 percentage points, at \((11, 4)\). Recruitment
of a Faithful player by the Traitors makes \(m_t\) private, so the Faithful no
longer know exactly when to switch to random voting. If only some of the
Faithful follow Vote-Left, their non-compliant votes can mask a Traitor's
deviation. The televised end game, in which the Faithful must banish every
remaining Traitor in a single vote, replaces the parity win condition, and
Traitors might deviate gradually or at random rather than switching cleanly at
\(n_t = 2m_t + 2\). Each of these is a natural extension of our analysis.

\section{Conclusion}

Vote-Left is a deterministic cyclic voting rule that reproduces the uniform
marginal of random voting while making any deviation observable, and it needs
nothing beyond the banishment vote itself. Combined with punishment, and with
the Faithful returning to random voting once \(n_t \leq 2m_t + 2\), it deters
every deviation in the main game, including joint deviations by the Traitors
(Theorem~\ref{thm:deterrence} and Proposition~\ref{prop:coalition}), and with a
single Traitor it is an equilibrium of the whole game
(Proposition~\ref{prop:single}). Across the configurations played on
television, this switching profile raises the Faithful's winning probability
by a factor of approximately four over random voting under collusion with
three Traitors, and by a factor of 62 to 84 with four.

What separates the two phases is timing: punishment deters a deviation only if
it arrives before the game is decided. In the late-game phase the Traitors do
best by colluding, and holding to Vote-Left there would hand them a certain
win. Random voting limits their gain to between 2.2 and 4.3 percentage points
over the compliance baseline at the televised configurations. Whether the
Faithful can recover this gain in an equilibrium of the full game remains open.

Operationally, the switching profile is simple: each player votes for the next
surviving player in the agreed cyclic ordering until the number of players
falls to \(2m_t + 2\), and votes at random thereafter. The robustness of the
protocol depends on the Faithful outnumbering the Traitors by more than two
through the main game, so that a detected deviator is punished before the
Traitors can reach parity; this is exactly the region of
Theorem~\ref{thm:deterrence}. Whether human players would actually adopt a
deterministic protocol over the in-room deliberation that television formats
encourage is an empirical question, and one that we suggest requires the
author to be given a spot on the show.

\section*{Data and code availability}

All results in this paper use \texttt{backstab} version 0.2.1, available from
PyPI and at \url{https://github.com/drvinceknight/backstab}. The data behind
every figure and table, with the scripts that generate them, are archived at
\url{https://doi.org/10.5281/zenodo.24528818}, and the source of the paper is at
\url{https://github.com/drvinceknight/traitors}.

\section*{Disclosure of interest}

The author reports no competing interests.

\section*{Funding}

No funding was received.

\bibliographystyle{plain}
\bibliography{refs}

\appendix

\section{Proof of Proposition~\ref{prop:collusion_algorithm}}
\label{app:algorithm}

For \(k = (k_0, \dots, k_{n-1})\) with \(\sum_j k_j = n - m\), the number of
profiles in \(\Omega\) in which player \(j\) receives \(k_j\) Faithful votes is
\[
N(k) = \sum_{S \subseteq F} (-1)^{|S|}\,
\frac{(n - m - |S|)!}{\prod_{j} (k_j - \mathbf{1}_{\{j \in S\}})!},
\]
with terms containing a negative factorial taken as zero. Indeed, without the
restriction \(v_i \neq i\) the number of such profiles is the multinomial
coefficient \((n - m)!/\prod_j k_j!\). If the voters in \(S\) are made to vote
for themselves, the remaining votes can be assigned in
\((n - m - |S|)!/\prod_j (k_j - \mathbf{1}_{\{j \in S\}})!\) ways, and
inclusion and exclusion over \(S\) removes every profile with a self-vote.

The summand of~\eqref{eq:coll_banish} depends on \(k\) only through the
maximum count \(M\), whether the target attains it, and the numbers of other
Faithful and of Traitors that attain it. Fix this pattern. Let \(x\) count
votes and \(z\) count members of \(S\). A Faithful player contributes the
polynomial \(\sum_{k} x^{k} (1/k! - z/(k - 1)!)\) and a Traitor contributes
\(\sum_{k} x^{k}/k!\), where each sum runs over the values of \(k\) that the
pattern allows. The sum of \(N(k)\) over the pattern is
\(\sum_s (n - m - s)!\) times the coefficient of \(x^{n - m} z^s\) in the
product of these polynomials.

Players of the same class and pattern contribute the same polynomial, so the
product is a product of powers of at most six polynomials of degree at most
\(n - m\). For each \(M\) the powers take \(O(n^3)\) operations to build. The
sums over \(M\), over the number of other Faithful attaining it and over \(s\)
each have at most \(n + 1\) terms, the sum over the number of Traitors
attaining it has \(m + 1\), and each term needs a truncated product of two
polynomials, which takes \(O(n^2)\) operations. The total is \(O(m n^5)\)
arithmetic operations.

\section{Simulation validation}\label{app:validation}

As an independent check of both the recurrences and their implementation, we
simulate the four profiles of Figure~\ref{fig:exact} directly with the game
engine of \texttt{backstab}, which plays out the votes, tie-breaks,
punishments and murders of the model. For every \(m \in \{2, 3, 4, 5\}\) and
\(n \leq 25\) we simulate \(100{,}000\) games with each of five random seeds and
pool them, giving \(500{,}000\) games for each case.
Figure~\ref{fig:validation} shows the difference between each simulated value
and the exact one. Leaving aside 21 cases in which the exact Faithful win
probability is positive but below \(10^{-5}\), too small for the simulation to
detect, the exact value lies inside the 95\% confidence interval of the
simulated one in 309 of 323 cases, close to the 95\% expected, and the largest
absolute difference is 0.0017. Measured in standard errors, the differences
have mean 0.07 and standard deviation 0.95, as expected if the simulation and
the recurrences agree.

\begin{figure}[tbp]
\centering
\includegraphics[width=0.8\textwidth]{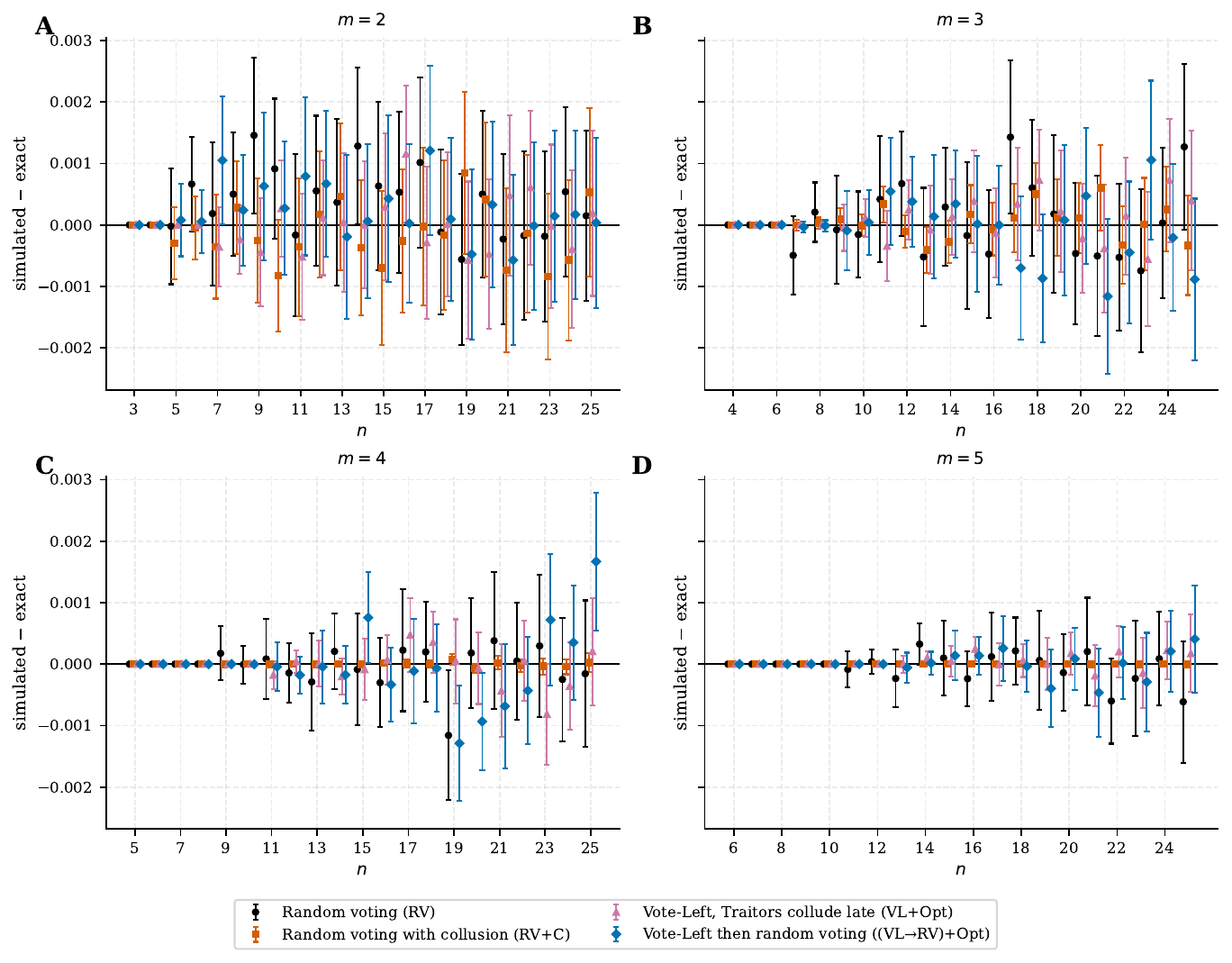}
\caption{\textbf{Simulation against the exact recurrences.} Each point is the
difference between the Faithful win probability over \(500{,}000\) simulated
games and the exact value of Figure~\ref{fig:exact}, for one profile and one
state \((n, m)\), with the 95\% confidence interval of the simulation as an
error bar: \textbf{A,} \(m = 2\); \textbf{B,} \(m = 3\); \textbf{C,}
\(m = 4\); \textbf{D,} \(m = 5\). The colours and markers are those of
Figure~\ref{fig:exact}, and the points for each \(n\) are offset slightly so
that they do not overlap. The intervals cross zero in all but a few cases, as
expected of 95\% intervals, and no difference exceeds 0.002.}
\label{fig:validation}
\end{figure}

\end{document}